\documentclass[conference,letterpaper]{IEEEtran}

\usepackage[utf8]{inputenc} 
\usepackage[T1]{fontenc}
\usepackage{url}
\usepackage{ifthen}
\usepackage{cite}
\usepackage[cmex10]{amsmath} 

\usepackage{graphicx} 
\usepackage{subcaption}
\usepackage{amsfonts}
\usepackage{amssymb}
\usepackage{amsthm}
\usepackage{natbib}
\usepackage{amsmath}
\usepackage{algorithm}
\usepackage{algorithmic}
\usepackage{xcolor}
\usepackage{url}
\usepackage{booktabs}
\usepackage{enumitem}
\usepackage{float}
\usepackage{hyperref}

\DeclareMathOperator{\tr}{tr}

\newcommand{\R}{\mathbb R}

\newcommand{\ip}[2]{\langle #1,#2\rangle}
\newcommand{\T}{\mathsf T}
\newcommand{\mmse}{\operatorname{mmse}}
\newcommand{\conv}{\operatorname{conv}}
\newcommand{\range}{\operatorname{range}}

\newcommand{\cI}{\mathcal I}
\newcommand{\cS}{\mathcal S}

\newcommand{\E}{\mathbb{E}}

\theoremstyle{plain}
\newtheorem{theorem}{Theorem}
\newtheorem{lemma}{Lemma}
\newtheorem{corollary}{Corollary}

\theoremstyle{definition}

\usepackage{todonotes}

\begin{document}


\title{Diffusion Sensitivity Geometry and Strong Duality for Optimal Privacy-Utility Tradeoffs} 

\author{%
\IEEEauthorblockN{Tue Do and Daniel Alabi}
\IEEEauthorblockA{Department of Electrical \& Computer Engineering\\
University of Illinois at Urbana-Champaign\\
Email: \{tuedo2,alabid\}@illinois.edu}%
}

\maketitle


\begin{abstract}
Existing approaches to differential privacy in diffusion models rely on loose sensitivity bounds and isotropic Gaussian noise, leading to unnecessary estimation error. We develop a framework for the optimal privacy-utility tradeoff of ideal diffusion denoisers, replacing conservative bounds with exact sensitivity characterizations and provably optimal noise calibration. First, we show that worst-case sensitivity over datasets of arbitrary size can be characterized exactly using only two records. For a data domain of size $m$, this reduces the analysis to at most $m(m-1)$ constraints, regardless of dataset size. Our characterization can be \textit{arbitrarily tighter than conventional diameter-based sensitivity bounds}, revealing potentially unbounded improvements in noise calibration. Second, we obtain exact privacy guarantees for R\'enyi, approximate, and Gaussian differential privacy, as well as mutual-information privacy for binary records, avoiding the looseness of previous generic privacy conversions. Finally, we show that the minimum achievable weighted mean-squared error can be computed through a semidefinite program and prove strong duality, providing certificates of global optimality.
\end{abstract}

\section{Introduction}

Diffusion models~\cite{ho2020denoising,song2020score}
now lead generative modeling of images and other high-dimensional data. Their training sets often contain sensitive content, revealed by the trained models; deep networks are shown to regenerate near-copies of training images~\cite{carlini2021extracting}, and an adversary can infer whether an individual record was used for training from the model's denoising behavior~\cite{rao2025score,do2026}. Differential privacy (DP)~\cite{DworkMNS06} is the standard remedy. Existing private diffusion models are trained with DP-SGD~\cite{abadi2016deep,dockhorn2023differentially,ghalebikesabi2023differentially}, which clips per-example gradients during training and adds isotropic Gaussian noise scaled to a fixed clipping norm. This calibration ignores the geometry of what is being protected and can add far more noise than the guarantee requires.

We ask what the optimal privacy-utility tradeoff is when privacy is enforced at the denoiser itself. We study explicitly the empirical risk minimizing denoiser (denoted as the \textit{ideal denoiser}), which materializes as a softmax-weighted average of the training records~\cite{do2026}. We release it through a Gaussian mechanism with a general covariance~$\Sigma$ and measure model utility by a weighted mean-squared error. We answer exactly how far one record can move the denoiser, how much privacy a given $\Sigma$ provides, and how to find the $\Sigma$ which provides the best privacy-utility tradeoff. Our contributions follow:

\begin{itemize}
  \item \emph{Sensitivity} (Section~\ref{sec:sensitivity}). The worst-case effect of one record, over datasets of every size, is generated by two-record datasets (Lemma~\ref{lem:convex-hull-sup}, Theorem~\ref{thm:twopoint}); for a fixed query and a domain of $m$ points, at most $m(m-1)$ constraints are needed. When one covariance must serve every query, the worst case is the largest Mahalanobis distance between two points of the domain (Corollary~\ref{cor:allqueries}), which can be far smaller than a naive diameter-based bound (Fig.~\ref{fig:simd10}).
  \item \emph{Privacy accounting} (Section~\ref{sec:renyi-dp}). We demonstrate that a single quantity $\rho_*(\Sigma)$ fixes the R\'enyi divergence between adjacent datasets at every order, the smallest zCDP and Gaussian DP parameters (Theorem~\ref{thm:rdp}), and the exact $(\varepsilon,\delta)$ curve for standard differential privacy, with no loss from generic conversions.
  \item \emph{Optimal noise} (Section~\ref{sec:primal-dual}). The least-error covariance under any of these budgets solves a semidefinite program whose conic dual has an exact gap formula; strong duality holds with both optima attained (Theorem~\ref{thm:dual}), so dual solutions certify global optimality.
  \item \emph{Mutual-information privacy} (Section~\ref{sec:mi-dp}). For binary records, mutual-information DP reduces to the capacity of a binary-input Gaussian channel whose capacity is strictly increasing in the signal-to-noise ratio (SNR) by the I-MMSE identity, giving another exact instance of the same program.
\end{itemize}

\begin{figure}
    \centering
    \includegraphics[width=1.0\linewidth]{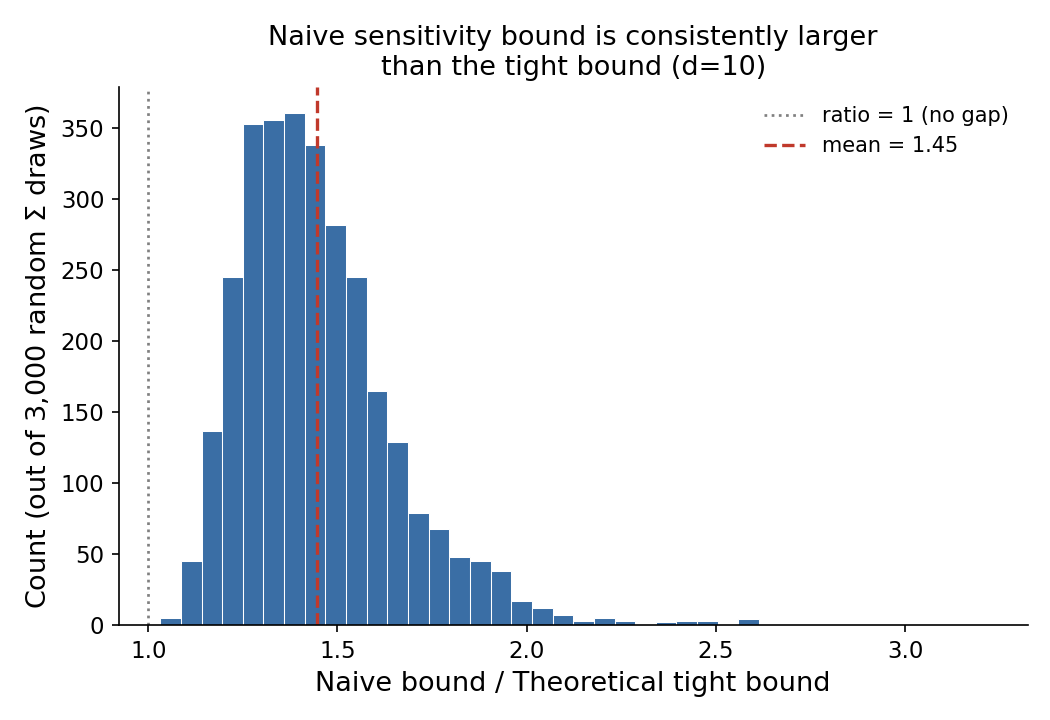}
    \caption{Histogram of the ratio of naive sensitivity bound $({\rm diam}(K)^2\cdot\lambda_{\max}(\Sigma^{-1}))$ to theoretical diffusion-specialized sensitivity bound over Gaussian mechanisms with random covariance matrices $\Sigma$ for a simulated convex hull $K=[0,1]^d$. Our theoretical bound achieves sensitivity analysis on average $1.45\times$ tighter with up to $2.5\times$ gain in tail cases.}
    \label{fig:simd10}
\end{figure}


\subsection{Related Work}

\emph{Privacy-utility tradeoffs.} Information theory has long studied the tension between disclosure and utility, as exemplified with Yamamoto's source coding problem with a secret component~\cite{yamamoto1983source}. It has since been studied through the privacy funnel~\cite{makhdoumi2014funnel}, MMSE estimation under privacy constraints~\cite{asoodeh2016mmse,asoodeh2019estimation}, and estimation guarantees under disclosure limits~\cite{wang2019privacy}. Under DP, Geng and Viswanath show that the optimal additive noise for a scalar query follows a staircase-shaped distribution rather than the Gaussian law~\cite{geng2016optimal}. However, none of these works directly apply to the diffusion model, which we show leads to interesting characterizations of privacy-utility tradeoffs.

\emph{Gaussian mechanisms and accounting.} We consider approximate DP~\cite{DworkMNS06}, R\'enyi DP~\cite{mironov2017renyi}, zCDP~\cite{bun2016concentrated}, Gaussian DP~\cite{dong2022}, and the mutual-information formulation of Cuff and Yu~\cite{CuffY16}; our exact $(\varepsilon,\delta)$ curve follows the analytic calibration of
Balle and Wang~\cite{balle2018gaussian}. Correlated Gaussian noise has been optimized for linear marginal queries through the matrix mechanism~\cite{xiao2023optimal} and for unbiased mean estimation~\cite{nikolov2024gaussian}. In those problems the sensitivity follows directly from the query; here the release is a specific nonlinear function of the data, and its sensitivity set must first be characterized. 

\emph{Information-estimation relations.} The I-MMSE
identity~\cite{guo2005mutual} has been extended to mismatched estimation~\cite{verdu2010mismatched}, arbitrary additive perturbations through score functions~\cite{guo2009relative}, Poisson and semimartingale channels~\cite{AW12,JVW18}, pointwise relations~\cite{VW12}, Gaussian networks and matrix-valued SNR~\cite{DBPS17,RPD18}, and channels with feedback and memory~\cite{HS14,HS16}. We use the identity to show that the binary-input channel induced by our mechanism has capacity strictly increasing in the SNR, which makes mutual-information privacy an exact sensitivity budget.

\section{The Ideal Denoiser}
\label{sec:ideal-denoiser}

\subsection{Normalized empirical posterior mean}

Let $K\subset\R^d$ be a public (nonempty) bounded domain. A dataset $D=(x_1,\ldots,x_n)$ is a (nonempty) finite multiset in $K$. Fix a public query $g\in\R^d$ and noise scale $\sigma>0$. Define
\begin{align}
 w_g(x)&=\exp\{-\|g-x\|^2/(2\sigma^2)\},\label{eq:weight}\\
 r_i(g;D)&=\frac{w_g(x_i)}{\sum_{j=1}^n w_g(x_j)},\qquad
 m_g(D)=\sum_{i=1}^n r_i(g;D)x_i.\label{eq:mean}
\end{align}
The denominator is a sum of the weights of $x_j$. If $X_0$ is drawn uniformly from $D$ and $G=X_0+\sigma Z$ with $Z\sim\mathcal N(0,I)$, Bayes' rule implies that
$m_g(D)=\E[X_0\mid G=g,D]$. In particular $m_g(D)\in\conv K$.

Write $G_t=\sqrt{a_t}X_0+\sqrt{1-a_t}Z$, where $0<a_t<1$ is the cumulative signal coefficient. 
For an ideal predictor $\epsilon^*$,\footnote{Replacing $\epsilon^*$ with a learned/nonideal denoiser results in a significantly different analysis.} we obtain
\begin{equation}\label{eq:diffusion}
 m_t(G_t;D)=\frac{G_t-\sqrt{1-a_t}\,\epsilon^*(G_t,t;D)}{\sqrt{a_t}}.
\end{equation}
The entire numerator is divided by $\sqrt{a_t}$. Equations~\eqref{eq:weight}--\eqref{eq:mean} use the rescaled query $g=G_t/\sqrt{a_t}$ and $\sigma^2=(1-a_t)/a_t$. The empirical posterior mean formulation is related to
the ideal diffusion noise predictor setting
studied in DIME \cite{do2026}.

\subsection{Mechanism and a linear objective}
Release only
\begin{equation}\label{eq:mechanism}
 M_{\Sigma,g}(D)=m_g(D)+Z_\Sigma,
 \qquad Z_\Sigma\sim\mathcal N(0,\Sigma).
\end{equation}
The covariance is the same for all private datasets, independent of the data and of the Gaussian draw. It may depend on public information and the fixed query. Let $L$ be a fixed linear map and $Q\succeq0$ its error weight. For $W=L^{\T}QL$, the expected additional quadratic loss is exactly
\begin{equation}\label{eq:utility}
 \E\|LZ_\Sigma\|_Q^2=\tr(W\Sigma).
\end{equation}
The cross term against any baseline reconstruction error has mean zero. Mean prediction MSE uses $W=I$. If noise is added to $\epsilon^*$ instead, the reconstruction cost is $(1-a_t)\tr(W\Sigma_\epsilon)/a_t$. Conversely, mean noise $Z_\Sigma$ is equivalent to noise $-\sqrt{a_t/(1-a_t)}Z_\Sigma$ in the ideal noise prediction.

We assume $W\succ0$ to obtain compactness and for ease of computing matrix square roots. A semidefinite task weight still defines a valid SDP, but may leave unpenalized directions with unbounded variance and an unattained infimum. Adding a specified small positive weight or separately bounding those directions changes the problem. Equation~\eqref{eq:utility} measures exact quadratic response or one-step reconstruction error. i.e., 
we do not use a nonlinear full sampler.

\section{From a Fixed Dataset to Global Sensitivity}
\label{sec:sensitivity}

\subsection{Exact deletion}

For $n\ge2$, let $D_{-k}$ delete record $k$. Define $s_k(g;D)=m_g(D)-m_g(D_{-k})$. Since all Gaussian weights are positive,
\begin{equation}\label{eq:delete}
 s_k=r_k(x_k-m_g(D_{-k}))
 =\frac{r_k}{1-r_k}(x_k-m_g(D)).
\end{equation}
Indeed, $m_g(D)=r_kx_k+(1-r_k)m_g(D_{-k})$. This shows that the relevant direction depends on the posterior responsibility and the location of the deleted record relative to the remaining mean. In particular,
\begin{equation}\label{eq:deletebound}
 \|s_k\|\le r_k\,\operatorname{diam}(K).
\end{equation}
The posterior spread also determines the query Jacobian:
\begin{equation}\label{eq:jacobian}
 \nabla_g m_g(D)=\sigma^{-2}\sum_i r_i(x_i-m_g(D))(x_i-m_g(D))^{\T},
\end{equation}
obtained by differentiating $r_i$. 

\subsection{Adjacency Notions}

\subsubsection{Add/Remove}

For global \emph{add/remove} privacy, our domain consists of all nonempty finite multisets in $K$. Define
\begin{equation}\label{eq:globalset}
 \cS_g^{\rm ar}=\{\pm s_k(g;D):D\subset K,\ |D|\ge2\}.
\end{equation}

\subsubsection{Replacement}

Define 
\begin{equation}\label{eq:globalsetrep}
 \cS_g^{\rm rep}=\{s_k(g;D\cup\{x_k\})-s_k(g;D\cup\{x_k'\}):D\subset K,|D|\geq 1\}
\end{equation}

For the replacement adjacency at fixed dataset size $N\ge1$, write neighboring datasets as $D=(u,a)$ and $D'=(u,b)$, and form the larger multiset $V=(u,a,b)$. Then
\begin{equation}
\begin{aligned}
 m_g(D)-m_g(D')={}&[m_g(D)-m_g(V)]\\
 &+[m_g(V)-m_g(D')].
\end{aligned}
\label{eq:triangle-mahalanobis}
\end{equation}
An add/remove Mahalanobis squared bound $\rho_{\rm ar}$ therefore implies a replacement bound $4\rho_{\rm ar}$ by the triangle inequality. 


\subsection{A finite exact reduction for one fixed query}

\begin{lemma}[\citep{rockafellar1997convex}]
    For any convex function $f$ and set $S$, it holds that
    $$\sup_S f=\sup_{\conv(S)}f.$$
    \label{lem:convex-hull-sup}
\end{lemma}

\begin{proof}
    For a convex $f$ and any $x\in\conv(S)$, write $x=\sum_i\lambda_i s_i$ with $s_i\in S,\lambda_i\geq 0,\sum_i\lambda_i=1$. Convexity gives us
    $f(x)\leq\sum_i\lambda_if(s_i)\leq\sup_{s\in S}f(s).$
    Taking $\sup$ over $x\in\conv(S)$ on the left, we get $\sup_{\conv(S)}f\leq\sup_S f$. The reverse inequality is trivial since $S\subseteq\conv(S)$.
\end{proof}

\begin{theorem}[Two record reduction]\label{thm:twopoint}
For fixed $g,\sigma$ and the add/remove adjacency notion, define
\begin{equation}\label{eq:twopoint}
 t_g(x,y)=\frac{w_g(x)}{w_g(x)+w_g(y)}(x-y),\qquad x,y\in K.
\end{equation}
The closed symmetric convex hull of $\cS_g^{\rm ar}$ equals that generated by all $t_g(x,y)$. Consequently, for every $\Sigma\succ0$,
\begin{equation}\label{eq:twopointsup}
 \sup_{s\in\cS_g^{\rm ar}}s^{\T}\Sigma^{-1}s
 =\sup_{x,y\in K}t_g(x,y)^{\T}\Sigma^{-1}t_g(x,y).
\end{equation}
\end{theorem}

\begin{proof}
    Denote $T=\{t_g(x,y):x,y\in K\}$. We first observe that $T\subseteq\cS_g^{\rm ar}$, so $T\subseteq\conv S$. Take $D=\{x,y\}$ valid 2-point dataset in $K$. Then
    $s_x(g;D)=m_g(\{x,y\})-m_g(\{y\})=t_g(x,y)$
    exactly, so every $t_g(x,y)$ itself is some $s_k(g;D)$. It remains to show that $S\subseteq\conv T$. Consider some arbitrary $s_k\in\cS^{\rm ar}_g$. We first show that $0\in\conv(T)$. Taking any $x,y\in K$, then clearly $0$ lies in the line between $t_g(x,y)=c_1(x-y)$ and $t_g(y,x)=c_2(y-x)$ for some $c_1,c_2>0$.
    We can expand
    \begin{align*}
        s_k&=\frac{r_k}{1-r_k}(x_k-m_g(D))=\frac{r_k}{1-r_k}\sum_{i\neq k}r_i(x_k-x_i) \\
        &=\sum_{i\neq k}\frac{r_ir_k}{1-r_k}(x_k-x_i)=\sum_{i\neq k}c_it_g(x_k,x_i)
    \end{align*}
    where 
    $$c_i:=\frac{r_ir_k(w_k+w_i)}{(1-r_k)w_k}=\frac{w_i(w_i+w_k)}{\bigl(\sum_i w_i\bigr)\bigl(\sum_{i\neq k}w_i\bigr)}\geq 0.$$
    Then, adding the constants up we have
    $$\sum_{i\neq k} c_i=\frac{\sum_{i\neq k} w_i^2+w_k\sum_{i\neq k}w_i}{\bigl(\sum_i w_i\bigr)\bigl(\sum_{i\neq k}w_i\bigr)}$$
    $$=\frac{\sum_{i\neq k}w_i^2}{\bigl(\sum_i w_i\bigr)\bigl(\sum_{i\neq k}w_i\bigr)}+\frac{w_k\sum_{i\neq k}w_i}{\bigl(\sum_i w_i\bigr)\bigl(\sum_{i\neq k}w_i\bigr)}$$
    $$\leq \frac{\bigl(\sum_{i\neq k}w_i\bigr)^2}{\bigl(\sum_i w_i\bigr)\bigl(\sum_{i\neq k}w_i\bigr)}+\frac{w_k}{\sum_i w_i}=1.$$
    where the last inequality comes from Jensen's inequality $(\sum_{i\neq k}w_i^2\leq[\sum_{i\neq k}w_i]^2)$.
    Thus, $s_k$ can be written as a convex combination of $t_g(x_k,x_i)$ for $i\neq k$ and $0$, which proves that $S\subseteq\conv T$. Eq.\ref{eq:twopointsup} follows directly from Lemma~\ref{lem:convex-hull-sup}.
\end{proof}

For a finite public domain of size $|K|=m$, at most $m(m-1)$ nonzero constraints suffice for an exact finite SDP, despite the unbounded dataset size. Explicitly, the constraint of dataset mean shifts can be represented with the $m(m-1)$ unique $t_g$ terms generated from finite domain $K$.

\begin{corollary}[Uniformly over all query inputs]\label{cor:allqueries}
If a single covariance must work for every $g\in\R^d$, then
\begin{equation}\label{eq:allg}
 \sup_{g,D,k}s_k(g;D)^{\T}\Sigma^{-1}s_k(g;D)
 =\sup_{x,y\in K}(x-y)^{\T}\Sigma^{-1}(x-y).
\end{equation}
\end{corollary}
\begin{proof}
    Clearly, $0$ is in both the convex hull the set $R=\{(x-y):x,y\in K\}$, and any element $t(x,y)\in T=\{t_g(x,y): x,y\in K\}$ is a convex combination of $0$ and $(x-y)$ in $R$, so $T\subseteq\conv R$, so from Theorem~\ref{thm:twopoint} we get that  $\sup_{g,D,k}s_k(g;D)^\top\Sigma^{-1}s_k(g;D)\leq\sup_{x,y\in K}(x-y)^\top\Sigma^{-1}(x-y).$
    Fix $x\neq y\in K$. Denote $g_\gamma=x+\gamma(x-y)$ for $\gamma>0$. Then it holds that
    $$t_{g_\gamma}(x,y)=\frac{1}{1+\exp\left(-\frac{(2\gamma+1)\Vert x-y\Vert^2}{2\sigma^2}\right)}(x-y).$$
    As $\gamma\rightarrow\infty$, $t_{g_\gamma}(x,y)\rightarrow(x-y)$. Since the quadratic product $v^T\Sigma^{-1}v$ is continuous, it holds that any product $(x-y)^\T\Sigma^{-1}(x-y)$ can be written as a limit of terms $t_g(x,y)^\T\Sigma^{-1}t_g(x,y)$, which by Theorem~\ref{thm:twopoint} is bounded above by $\sup_{g,D,K}s_k^\T\Sigma^{-1}s_k$. Since limits preserve inequality, it holds that
$\sup_{g,D,k}s_k(g;D)^\top\Sigma^{-1}s_k(g;D)\geq\sup_{x,y\in K}(x-y)^\top\Sigma^{-1}(x-y).$
\end{proof}

\section{R\'enyi DP and Related Gaussian Privacy Definitions}
\label{sec:renyi-dp}

For a chosen adjacency relation, write $P_D$ for the output distribution and let $\cS$ be the full neighboring mean difference set (e.g., Equation~\eqref{eq:globalset}), or a set with the same closed symmetric convex hull. For $\alpha>1$,
\begin{equation}\label{eq:rdpdef}
 D_\alpha(P\|Q)=\frac{1}{\alpha-1}
 \log\int\left(\frac{dP}{dQ}\right)^\alpha dQ.
\end{equation}
A mechanism is $(\alpha,\varepsilon_\alpha)$-RDP if $D_\alpha(P_D\|P_{D'})\le\varepsilon_\alpha$ for every ordered neighboring pair \cite{mironov2017renyi}. It is $(\varepsilon,\delta)$-DP if $P_D(A)\le e^\varepsilon P_{D'}(A)+\delta$ for every such pair and measurable $A$ \cite{DworkMNS06}.

\begin{theorem}[Exact Gaussian privacy curve]\label{thm:rdp}
For a common positive definite covariance, set
$\rho_*(\Sigma)=\sup_{s\in\cS}s^{\T}\Sigma^{-1}s.$
The smallest valid RDP parameter at every order is
\begin{equation}\label{eq:rdpcurve}
 \varepsilon_\alpha^*(\Sigma)=\frac\alpha2\rho_*(\Sigma).
\end{equation}
Thus, one order-$\alpha$ budget is equivalent to $\rho_*\le2\varepsilon_\alpha/\alpha$, and several budgets are equivalent to $\rho_*\le\inf_{\alpha\in\mathcal A}2\varepsilon_\alpha/\alpha$, where $\mathcal A$ is the collection of R\'enyi orders being considered. The same mechanism is exactly $(\rho_*/2)$-zCDP and $\sqrt{\rho_*}$-GDP in the sense of their smallest parameters \cite{bun2016concentrated,dong2022}.
\end{theorem}

\begin{proof}
    The R\`enyi divergence between two Gaussian distributions with common covariance matrix has closed form~\citep{van2014renyi,mironov2017renyi}
    $$D_\alpha(\mathcal{N}(\mu_1,\Sigma)\Vert\mathcal{N}(\mu_2,\Sigma))=\frac{\alpha}{2}(\mu_1-\mu_2)^\T\Sigma^{-1}(\mu_1-\mu_2),$$
    for all orders $\alpha$. Thus, the R\'enyi divergence for specific $\alpha$ between noised outputs of adjacent datasets has lowest upper bound of $\frac{\alpha}{2}\rho_*(\Sigma)$.
    A mechanism satisfies $\rho$-zCDP if and only if its R\`enyi Divergence between neighboring outputs is bounded by $\rho\alpha$ for every $\alpha>1$~\citep{bun2016concentrated}, which means we directly get that the mechanism satisfies $(\rho_*/2)$-zCDP. GDP measures privacy by how well the best possible hypothesis test can distinguish a mechanism's output on two neighboring datasets, and says a mechanism is $\mu$-GDP if that best test is never more powerful than distinguishing $\mathcal{N}(0,1)$ and $\mathcal{N}(\mu,1)$. Testing $\mathcal{N}(\mu_1,\Sigma)$ vs. $\mathcal{N}(\mu_2,\Sigma)$ is equivalent by invertible whitening map $\Sigma^{-1/2}$ to testing two isotropic Gaussians separated by $\Vert\Sigma^{-1/2}(\mu_1-\mu_2)\Vert$. The worst case neighbor mean shift is captured exactly by $\sqrt{\rho_*(\Sigma)}$ so the mechanism satisfies $\sqrt{\rho_*}$-GDP.
\end{proof}

A general RDP guarantee implies $(\varepsilon_\alpha+\frac{1}{\alpha-1}\log(1/\delta)/(\alpha-1),\delta)$-DP for $0<\delta<1$~\cite{mironov2017renyi}. Define $\Lambda:=\log\frac{p(x)}{q(x)},$ regarded as a log-likelihood random variable for two Gaussian laws $p,q$ representing mechanisms characterized by adjacent datasets.
Optimizing this conversion over \eqref{eq:rdpcurve} gives the budget
\begin{equation}\label{eq:zcdpconversion}
 \varepsilon=\rho_*/2+\sqrt{2\rho_*\log(1/\delta)}.
\end{equation}

Let $\Phi$ the standard Gaussian CDF. For $P=\mathcal{N}(s_j,\Sigma)$ and $Q=\mathcal{N}(0,\Sigma)$, the log-likelihood ratio $\Lambda$ is affine in a Gaussian, so for $q:=s_j^\T\Sigma^{-1}s_j$ it holds that
$\Lambda|P\sim\mathcal{N}(q/2,q),\qquad\Lambda|Q\sim\mathcal{N}(-q/2,q).$
Each tail probability is then a single Gaussian CDF:
$${\textstyle P(\Lambda>\varepsilon)=\Phi\left(\frac{\sqrt{q}}{2}-\frac{\varepsilon}{\sqrt{q}}\right),\;\; Q(\Lambda>\varepsilon)=\Phi\left(-\frac{\sqrt{q}}{2}-\frac{\varepsilon}{\sqrt{q}}\right)}.$$
$(\varepsilon,\delta)$-DP requires that $\delta\geq\sup_A[P(A)-e^\varepsilon Q(A)]$, so we define
\begin{align}\label{eq:profile}
 h_\varepsilon(q)&:=\Phi(\sqrt q/2-\varepsilon/\sqrt q)\nonumber\\
 &\quad-e^\varepsilon\Phi(-\sqrt q/2-\varepsilon/\sqrt q),\quad
 h_\varepsilon(0)=0.
\end{align}

The maximizing event $A$ in $P(A)-e^\varepsilon Q(A)$ is $\{\Lambda>\varepsilon\}$. For $\varepsilon\ge0$, $h_\varepsilon$ increases continuously from zero to one, since $h_\varepsilon'(q)=\frac{\phi(\sqrt q/2-\varepsilon/\sqrt q)}{2\sqrt q}>0.$
Thus, the exact worst-case $\delta$ is $h_\varepsilon(\rho_*)$ \cite{balle2018gaussian}. 
We conclude that $(\varepsilon,\delta)$-DP is exactly given by $\rho_*\le h_\varepsilon^{-1}(\delta)$ when $0<\delta<1$. For any nonzero shift, $\Lambda$ has an unbounded upper 
tail, i.e., finite Gaussian covariance cannot give pure DP at finite $\varepsilon$.

\section{The Covariance Primal and Its Conic Dual}
\label{sec:primal-dual}

Let $s_1,\ldots,s_J\in\R^d$ describe the full symmetric convex hull of a specified sensitivity set. Take $\rho>0$, $W\succ0$, and a public noise floor $\eta\ge0$. All matrices are real 
and symmetric, and $\ip{A}{B}=\tr(A^{\T}B)$. 

The primal is
\begin{equation}\tag{$\mathrm P_{\rho,\eta}$}\label{eq:primal}
\begin{aligned}
 v_\eta(\rho)=\min_{\Sigma\in\mathbb S^d}\;&\ip{W}{\Sigma}\,\,\,
 \text{s.t. }\Sigma-\eta I\succeq0,\\
 &F_j(\Sigma):=\begin{bmatrix}\Sigma&s_j\\s_j^{\T}&\rho\end{bmatrix}\succeq0,
 \quad j=1,\ldots,J.
\end{aligned}
\end{equation}
For $\eta>0$ the ordinary Schur complement gives $s_j^{\T}\Sigma^{-1}s_j\le\rho$. For $\eta=0$, the generalized Schur complement gives both $s_j\in\range\Sigma$ and $s_j^{\T}\Sigma^\dagger s_j\le\rho$. Equivalently, $\Sigma\succeq s_js_j^{\T}/\rho$. 

With $\rho=2\varepsilon_\alpha/\alpha$, \eqref{eq:primal} exactly minimizes the loss among the specified common covariance Gaussian mechanisms satisfying RDP when the sensitivity hull is exact. Theorem~\ref{thm:twopoint} makes a finite public $K$ exact for a single fixed query; Corollary~\ref{cor:allqueries} handles a uniform-query covariance.

\begin{theorem}[Conic dual, gap, and strong duality]\label{thm:dual}
Introduce
\[
 Z_j=\begin{bmatrix}A_j&b_j\\b_j^{\T}&\gamma_j\end{bmatrix}\succeq0,
 \qquad H\succeq0.
\]
A dual of \eqref{eq:primal} is
\begin{equation}\tag{$\mathrm D_{\rho,\eta}$}\label{eq:dual}
\begin{aligned}
 d_\eta(\rho)=\max\;&\eta\tr H-2\sum_jb_j^{\T}s_j-\rho\sum_j\gamma_j\\
 \text{s.t. }&\sum_j A_j+H=W,\quad Z_j\succeq0,\quad H\succeq0.
\end{aligned}
\end{equation}
For any feasible pair, its objective gap equals
\begin{equation}\label{eq:gap}
 \sum_j\ip{Z_j}{F_j(\Sigma)}+\ip{H}{\Sigma-\eta I}\ge0.
\end{equation}
For finite $J$, $W\succ0$, $\rho>0$, and finite shifts, both programs attain their optima and their values agree.
\end{theorem}
\begin{proof}
    We first show a block-trace identity. Block multiplication gives
    $$Z_jF_j(\Sigma)=\begin{bmatrix}
        A_j\Sigma+b_js_j^\T & A_js_j+b_js_j^\T \\
        b^\T_j\Sigma + \gamma_js_j^\T & b_j^\T+\gamma_j\rho
    \end{bmatrix},$$
    so
    $\langle Z_j,F_j(\Sigma)\rangle=\langle A_j,\Sigma\rangle+2b_j^\T s_j+\gamma_j\rho.$
    We dualize the constraints $F_j(\Sigma)\succeq 0$ each with $Z_j\succeq 0$ and $\Sigma-\eta I\succeq 0$ with $H\succeq 0$ and, giving us the Lagrangian
    $$\mathcal{L}(\Sigma;H,\{Z_j\})=\langle W,\Sigma\rangle-\sum_j\langle Z_j,F_j(\Sigma)\rangle-\langle H,\Sigma-\eta I\rangle.$$
    We can simplify using our block-trace identity
    $$\mathcal{L}=\left\langle W-H-\sum_j A_j,\Sigma\right\rangle+\eta\tr H-2\sum_j b_j^\T s_j-\rho\sum_j\gamma_j.$$
    The dual function is $g(H,\{Z_j\})=\inf_\Sigma\mathcal{L}(\Sigma;H,\{Z_j\})$. However, in the dual function $\Sigma$ is unconstrained so unless the coefficient with $\Sigma$ is zero the infimum is $-\infty$, giving us the constraint $\sum_jA_j+H=W.$
    When the constraint holds,
    $g(H,\{Z_j\})=\eta\tr H-2\sum_j b_j^\T s_j-\rho\sum_j\gamma_j,$
    demonstrating that~\eqref{eq:dual} is the Lagrangian dual of~\eqref{eq:primal}. Let $\Sigma$ be primal-feasible and $(H,\{Z_j\})$ be dual-feasible, so in particular $H+\sum_j A_j=W$. We get that
    $\langle W, \Sigma\rangle=\langle H,\Sigma\rangle +\sum_j\langle A_j,\Sigma\rangle,$
    writing $\langle H,\Sigma\rangle=\langle H,\Sigma-\eta I\rangle+\eta\tr H$ and using the block-trace identity we write $\langle W,\Sigma\rangle$ as:
    $$\langle H,\Sigma-\eta I\rangle+\eta\tr H+\sum_j\langle Z_j,F_j(\Sigma)\rangle-2\sum_j b_j^\T s_j-\rho\sum_j\gamma_j.$$
    Moving the dual-objective terms to the left-hand side gives us~\eqref{eq:gap} for any feasible pair. We are left with the strong duality claim. Slater's condition tells us that in a convex program if there exists a point in the primal with strict feasibility and the primal has boundedness from below then the primal and dual optimum exist and agree~\citep{vandenberghe2004convex}. Since $W\succ 0$ and $\Sigma\succeq 0$ clearly $v_\eta(\rho)\geq 0$, so we have boundedness from below. Consider $\Sigma=cI$ for some constant $c>\max\{\eta,\max_j\Vert s_j\Vert^2/\rho\}$. Clearly $\Sigma-\eta I\succ 0$ and $F_j(\Sigma)\succ 0$ for all $j$ so we have strict feasibility. $\langle W,\Sigma\rangle\rightarrow\infty$ as $\Vert \Sigma\Vert\rightarrow\infty$ along any direction in the closed feasible set, so we conclude that the optimum is attainable by Extreme Value Theorem~\citep{rudin2021principles}.
\end{proof}

\section{MI-DP: Exact Programs}
\label{sec:mi-dp}

\subsection{Definition and conditional channel capacity}
Fix a database size $N$ and product domain $\mathcal X=\prod_i\mathcal X_i$ (for the denoiser, $K^N$). For random $X$, Cuff--Yu MI-DP is
\begin{equation}\label{eq:midpdef}
 \cI(M):=\sup_{P_X}\max_i I(X_i;M(X)\mid X_{-i})\le\kappa.
\end{equation}

All information is in nats. The supremum includes \emph{all joint priors}, with arbitrary record dependence \cite{cuff2016mutual}. 

Let $\beta=(i,u)$
where $i$ is the coordinate/record whose privacy we are measuring and $u$ is the fixed value of all the other records
$X_{-i}$. Let $W_\beta$ be the channel $a\mapsto M(u,a)$. Here $a\in\mathcal{X}_i$ is the channel input, i.e. the value of record $i$, and the output is $M(u,a)$, the mechanism run on the dataset with record $i$ equal to $a$ and all other records fixed by $u$. Let $\pi$ be a probability distribution on $\mathcal{X}_i$, and let $A\sim\pi$ and $Y:=M(u,a)$.
Then
\begin{equation}\label{eq:capacity}
 \cI(M)=\sup_\beta\mathsf C(W_\beta),\qquad
 \mathsf C(W_\beta)=\sup_\pi I_\pi(A;Y).
\end{equation}
This is easy to see if one conditions any prior on $X_{-i}\triangleq u$: the conditional information is an average of channel informations, each at most the corresponding capacity. Conversely, choose $X_{-i}$ deterministically and any distribution $\pi$ for the remaining record. Taking suprema leads to \eqref{eq:capacity}.

\subsection{Exact binary record calibration and its dual}

Suppose every record has two possible values. Then when the Gaussian means differ by $s$, we let $q=s^{\T}\Sigma^{-1}s$. Whitening and translating reduces its (informative) coordinate to
\[
 Y=\sqrt q\,B/2+Z,\quad Z\sim\mathcal N(0,1),\quad B\in\{-1,1\}.
\]
By reflection symmetry and concavity of mutual information in the input distribution, we know
that the uniform sign is capacity-achieving. Its capacity is
\begin{equation}\label{eq:binarycapacity}
 \psi(q)=\log2-\E_Z\log\left(1+e^{-q/2-\sqrt q Z}\right).
\end{equation}
The real Gaussian I--MMSE identity \cite{guo2005mutual} gives
\begin{equation}\label{eq:immse}
 \psi'(q)=\frac18\mmse(B\mid\sqrt q\,B/2+Z)>0.
\end{equation}
Also $\psi(0)=0$ and $\psi(q)\to\log2$ as $q\to\infty$: a sign detector has vanishing error, so its conditional entropy vanishes. Thus, for $0<\kappa<\log2$,
\begin{equation}\label{eq:exactbinary}
 \kappa\text{-MI-DP}\quad\Longleftrightarrow\quad
 \sup_{s\in\cS^{\rm rep}}s^{\T}\Sigma^{-1}s\le\psi^{-1}(\kappa).
\end{equation}
This gives an exact instance of \eqref{eq:primal} and \eqref{eq:dual} using the \emph{full replacement} sensitivity set. For singular covariance define $q_\Sigma(s)=s^{\T}\Sigma^\dagger s$ when $s\in\range\Sigma$, and $q_\Sigma(s)=+\infty$ otherwise. The same equivalence uses $q_\Sigma$ in place of the inverse quadratic: disjoint supports give $\psi(+\infty)=\log2$, and common supports reduce to the preceding Gaussian channel. For add/remove adjacency, the choice of $4\rho_{\rm ar}\le\psi^{-1}(\kappa)$ suffices for $N\ge1$, following directly from the triangle inequality in Equation~\eqref{eq:triangle-mahalanobis}.



\section{Conclusion}

We characterized the optimal privacy-utility tradeoff for releasing the ideal diffusion denoiser through a Gaussian mechanism with a general covariance. The worst-case sensitivity over datasets of any size is generated by pairs of records, specifically the largest Mahalanobis distance between two points of the data domain. This sensitivity quantity fixes the mechanism's R\'enyi, zCDP, Gaussian, and approximate DP guarantees exactly, as well as its mutual-information privacy for binary records. Finding the least-error covariance under any of these budgets is a semidefinite program whose strong dual certifies global optimality. Together, these results replace conservative sensitivity bounds with exact privacy accounting and provably optimal noise design for privacy-preserving diffusion denoisers.

\clearpage

\small{
\bibliographystyle{unsrt}
\bibliography{references}

@article{ho2020denoising,
  title={Denoising diffusion probabilistic models},
  author={Ho, Jonathan and Jain, Ajay and Abbeel, Pieter},
  journal={Advances in neural information processing systems},
  volume={33},
  pages={6840--6851},
  year={2020}
}

@article{song2020score,
  title={Score-based generative modeling through stochastic differential equations},
  author={Song, Yang and Sohl-Dickstein, Jascha and Kingma, Diederik P and Kumar, Abhishek and Ermon, Stefano and Poole, Ben},
  journal={arXiv preprint arXiv:2011.13456},
  year={2020}
}

@article{ghalebikesabi2023differentially,
  title={Differentially private diffusion models generate useful synthetic images},
  author={Ghalebikesabi, Sahra and Berrada, Leonard and Gowal, Sven and Ktena, Ira and Stanforth, Robert and Hayes, Jamie and De, Soham and Smith, Samuel L and Wiles, Olivia and Balle, Borja},
  journal={arXiv preprint arXiv:2302.13861},
  year={2023}
}

@article{
dockhorn2023differentially,
title={Differentially Private Diffusion Models},
author={Tim Dockhorn and Tianshi Cao and Arash Vahdat and Karsten Kreis},
journal={Transactions on Machine Learning Research},
issn={2835-8856},
year={2023},
url={https://openreview.net/forum?id=ZPpQk7FJXF},
note={}
}

@inproceedings{DworkMNS06,
  author    = {Cynthia Dwork and
               Frank McSherry and
               Kobbi Nissim and
               Adam D. Smith},
  title     = {Calibrating Noise to Sensitivity in Private Data Analysis},
  booktitle = {TCC},
  year      = {2006}
}

@inproceedings{carlini2021extracting,
  title={Extracting training data from large language models},
  author={Carlini, Nicholas and Tramer, Florian and Wallace, Eric and Jagielski, Matthew and Herbert-Voss, Ariel and Lee, Katherine and Roberts, Adam and Brown, Tom and Song, Dawn and Erlingsson, Ulfar and others},
  booktitle={30th USENIX security symposium (USENIX Security 21)},
  pages={2633--2650},
  year={2021}
}

@inproceedings{CuffY16,
  author       = {Paul Cuff and
                  Lanqing Yu},
  title        = {Differential Privacy as a Mutual Information Constraint},
  booktitle    = {Proceedings of the 2016 {ACM} {SIGSAC} Conference on Computer and
                  Communications Security, Vienna, Austria, October 24-28, 2016},
  pages        = {43--54},
  publisher    = {{ACM}},
  year         = {2016}
}

@article{guo2005mutual,
  title={Mutual information and minimum mean-square error in Gaussian channels},
  author={Guo, Dongning and Shamai, Shlomo and Verd{\'u}, Sergio},
  journal={IEEE transactions on information theory},
  volume={51},
  number={4},
  pages={1261--1282},
  year={2005},
  publisher={IEEE}
}

@article{verdu2010mismatched,
  title={Mismatched estimation and relative entropy},
  author={Verd{\'u}, Sergio},
  journal={IEEE Transactions on Information Theory},
  volume={56},
  number={8},
  pages={3712--3720},
  year={2010},
  publisher={IEEE}
}

@inproceedings{guo2009relative,
  title={Relative entropy and score function: New information-estimation relationships through arbitrary additive perturbation},
  author={Guo, Dongning},
  booktitle={2009 IEEE International Symposium on Information Theory},
  pages={814--818},
  year={2009},
  organization={IEEE}
}

@ARTICLE{JVW18,
  author={Jiao, Jiantao and Venkat, Kartik and Weissman, Tsachy},
  journal={IEEE Transactions on Information Theory}, 
  title={Mutual Information, Relative Entropy and Estimation Error in Semi-Martingale Channels}, 
  year={2018},
  volume={64},
  number={10},
  pages={6662-6671}
}

@ARTICLE{AW12,
  author={Atar, Rami and Weissman, Tsachy},
  journal={IEEE Transactions on Information Theory}, 
  title={Mutual Information, Relative Entropy, and Estimation in the Poisson Channel}, 
  year={2012},
  volume={58},
  number={3},
  pages={1302-1318}
}

@INPROCEEDINGS{VW12,
  author={Venkat, Kartik and Weissman, Tsachy},
  booktitle={2012 IEEE International Symposium on Information Theory Proceedings}, 
  title={Pointwise relations between information and estimation in Gaussian noise}, 
  year={2012},
  volume={},
  number={},
  pages={701-705},
  doi={10.1109/ISIT.2012.6284310}
}

@INPROCEEDINGS{RPD18,
  author={Reeves, Galen and Pfister, Henry D. and Dytso, Alex},
  booktitle={2018 IEEE International Symposium on Information Theory (ISIT)}, 
  title={Mutual Information as a Function of Matrix SNR for Linear Gaussian Channels}, 
  year={2018},
  volume={},
  number={},
  pages={1754-1758},
  doi={10.1109/ISIT.2018.8437326}}

@Article{DBPS17,
AUTHOR = {Dytso, Alex and Bustin, Ronit and Poor, H. Vincent and Shamai (Shitz), Shlomo},
TITLE = {A View of Information-Estimation Relations in Gaussian Networks},
JOURNAL = {Entropy},
VOLUME = {19},
YEAR = {2017},
NUMBER = {8},
ARTICLE-NUMBER = {409},
URL = {https://www.mdpi.com/1099-4300/19/8/409},
ISSN = {1099-4300},
DOI = {10.3390/e19080409}
}

@INPROCEEDINGS{HS14,
  author={Han, Guangyue and Song, Jian},
  booktitle={2014 IEEE International Symposium on Information Theory}, 
  title={Extensions of the I-MMSE relation}, 
  year={2014},
  volume={},
  number={},
  pages={2202-2206}
}

@ARTICLE{HS16,
  author={Han, Guangyue and Song, Jian},
  journal={IEEE Transactions on Information Theory}, 
  title={Extensions of the I-MMSE Relationship to Gaussian Channels With Feedback and Memory}, 
  year={2016},
  volume={62},
  number={10},
  pages={5422-5445}
}

@article{do2026,
  author = {Tue Do and Daniel Alabi},
  title = {{DIME}: Query-Efficient Framework for Membership Inference on Diffusion Models},
  journal={arXiv preprint arXiv:2608.22824},
  year = {2026},
  eprint = {2608.22824},
  archivePrefix = {arXiv},
  primaryClass = {cs.LG},
  url = {https://arxiv.org/abs/2608.22824}
}

@article{dong2022,
 author = {Jinshuo Dong and Aaron Roth and Weijie J. Su},
 title = {{Gaussian} Differential Privacy},
 journal = {Journal of the Royal Statistical Society: Series B (Statistical Methodology)},
 volume = {84},
 number = {1},
 pages = {3--37},
 year = {2022},
 doi = {10.1111/rssb.12454}
}

@inproceedings{balle2018gaussian,
  author = {Balle, Borja and Wang, Yu-Xiang},
  title = {Improving the {Gaussian} Mechanism for Differential Privacy: Analytical Calibration and Optimal Denoising},
  booktitle = {Proceedings of the 35th International Conference on Machine Learning},
  editor = {Dy, Jennifer and Krause, Andreas},
  series = {Proceedings of Machine Learning Research},
  volume = {80},
  pages = {394--403},
  publisher = {PMLR},
  year = {2018},
  url = {https://proceedings.mlr.press/v80/balle18a.html}
}

@inproceedings{bun2016concentrated,
  author = {Bun, Mark and Steinke, Thomas},
  title = {Concentrated Differential Privacy: Simplifications, Extensions, and Lower Bounds},
  booktitle = {Theory of Cryptography (TCC 2016-B), Part I},
  editor = {Hirt, Martin and Smith, Adam},
  series = {Lecture Notes in Computer Science},
  volume = {9985},
  pages = {635--658},
  publisher = {Springer},
  year = {2016},
  doi = {10.1007/978-3-662-53641-4_24},
  url = {https://arxiv.org/abs/1605.02065}
}

@inproceedings{mironov2017renyi,
  author = {Mironov, Ilya},
  title = {{R\'enyi} Differential Privacy},
  booktitle = {2017 IEEE 30th Computer Security Foundations Symposium (CSF)},
  pages = {263--275},
  publisher = {IEEE},
  year = {2017},
  doi = {10.1109/CSF.2017.11},
  url = {https://arxiv.org/abs/1702.07476}
}

@inproceedings{cuff2016mutual,
  author = {Cuff, Paul and Yu, Lanqing},
  title = {Differential Privacy as a Mutual Information Constraint},
  booktitle = {Proceedings of the 2016 ACM SIGSAC Conference on Computer and Communications Security},
  year = {2016},
  pages = {43--54},
  publisher = {ACM},
  doi = {10.1145/2976749.2978308},
  url = {https://arxiv.org/abs/1608.03677}
}

@book{rockafellar1997convex,
  title={Convex analysis},
  author={Rockafellar, R Tyrrell},
  volume={28},
  year={1997},
  publisher={Princeton university press}
}

@article{van2014renyi,
  title={R{\'e}nyi divergence and Kullback-Leibler divergence},
  author={Van Erven, Tim and Harremos, Peter},
  journal={IEEE transactions on information theory},
  volume={60},
  number={7},
  pages={3797--3820},
  year={2014},
  publisher={IEEE}
}

@book{vandenberghe2004convex,
  title={Convex optimization},
  author={Vandenberghe, Lieven and Boyd, Stephen},
  volume={1},
  year={2004},
  publisher={Cambridge university press Cambridge}
}

@article{yamamoto1983source,
  author = {Hirosuke Yamamoto},
  title = {A Source Coding Problem for Sources with Additional Outputs to Keep Secret from the Receiver or Wiretappers},
  journal = {IEEE Transactions on Information Theory},
  year = {1983},
  volume = {29},
  number = {6},
  pages = {918--923},
  month = nov,
  doi = {10.1109/TIT.1983.1056749}
}

@inproceedings{makhdoumi2014funnel,
  author = {Ali Makhdoumi and Salman Salamatian and Nadia Fawaz and Muriel M{\'e}dard},
  title = {From the Information Bottleneck to the Privacy Funnel},
  booktitle = {Proceedings of the IEEE Information Theory Workshop (ITW)},
  year = {2014},
  pages = {501--505},
  doi = {10.1109/ITW.2014.6970882},
  url = {https://arxiv.org/abs/1402.1774}
}

@inproceedings{asoodeh2016mmse,
  author = {Shahab Asoodeh and Fady Alajaji and Tam{\'a}s Linder},
  title = {Privacy-Aware {MMSE} Estimation},
  booktitle = {Proceedings of the IEEE International Symposium on Information Theory (ISIT)},
  year = {2016},
  pages = {1989--1993},
  doi = {10.1109/ISIT.2016.7541647},
  url = {https://arxiv.org/abs/1601.07417}
}

@article{asoodeh2019estimation,
  author = {Shahab Asoodeh and Mario Diaz and Fady Alajaji and Tam{\'a}s Linder},
  title = {Estimation Efficiency Under Privacy Constraints},
  journal = {IEEE Transactions on Information Theory},
  year = {2019},
  volume = {65},
  number = {3},
  pages = {1512--1534},
  month = mar,
  doi = {10.1109/TIT.2018.2865558},
  url = {https://arxiv.org/abs/1707.02409}
}

@article{wang2019privacy,
  author = {Hao Wang and Lisa Vo and Fl{\'a}vio P. Calmon and Muriel M{\'e}dard and Ken R. Duffy and Mayank Varia},
  title = {Privacy With Estimation Guarantees},
  journal = {IEEE Transactions on Information Theory},
  year = {2019},
  volume = {65},
  number = {12},
  pages = {8025--8042},
  month = dec,
  doi = {10.1109/TIT.2019.2934414},
  url = {https://arxiv.org/abs/1710.00447}
}

@article{geng2016optimal,
  author = {Geng, Quan and Viswanath, Pramod},
  title = {The Optimal Noise-Adding Mechanism in Differential Privacy},
  journal = {IEEE Transactions on Information Theory},
  volume = {62},
  number = {2},
  pages = {925--951},
  month = feb,
  year = {2016},
  doi = {10.1109/TIT.2015.2504967},
  url = {https://pramodv.web.engr.illinois.edu/pubs/GV16_1.pdf}
}

@inproceedings{xiao2023optimal,
  author = {Xiao, Yingtai and He, Guanlin and Zhang, Danfeng and Kifer, Daniel},
  title = {An Optimal and Scalable Matrix Mechanism for Noisy Marginals under Convex Loss Functions},
  booktitle = {Advances in Neural Information Processing Systems},
  volume = {36},
  year = {2023},
  url = {https://proceedings.neurips.cc/paper_files/paper/2023/hash/414f4c9fe9653e5de98fad6964d50315-Abstract-Conference.html}
}

@inproceedings{nikolov2024gaussian,
  author = {Nikolov, Aleksandar and Tang, Haohua},
  title = {General {Gaussian} Noise Mechanisms and Their Optimality for Unbiased Mean Estimation},
  booktitle = {15th Innovations in Theoretical Computer Science Conference (ITCS 2024)},
  series = {Leibniz International Proceedings in Informatics (LIPIcs)},
  volume = {287},
  pages = {85:1--85:23},
  publisher = {Schloss Dagstuhl -- Leibniz-Zentrum f{\"u}r Informatik},
  year = {2024},
  doi = {10.4230/LIPIcs.ITCS.2024.85}
}

@article{rudin2021principles,
  title={Principles of mathematical analysis},
  author={Rudin, Walter},
  year={2021}
}

@inproceedings{abadi2016deep,
  author    = {Mart{\'\i}n Abadi and Andy Chu and Ian Goodfellow and
               H. Brendan McMahan and Ilya Mironov and Kunal Talwar and
               Li Zhang},
  title     = {Deep Learning with Differential Privacy},
  booktitle = {Proceedings of the 2016 ACM SIGSAC Conference on Computer and
               Communications Security},
  pages     = {308--318},
  year      = {2016}
}

@article{rao2025score,
  title={Score-based membership inference on diffusion models},
  author={Rao, Mingxing and Qu, Bowen and Moyer, Daniel},
  journal={arXiv preprint arXiv:2509.25003},
  year={2025}
}
}

\section{AI Use Acknowledgment}
In this work, we used generative AI tools to generate
and optimize the code that was used 
to run the simulation and experiments (Fig.~\ref{fig:simd10}).

\end{document}